\documentclass[lettersize,journal]{IEEEtran}

\ifCLASSINFOpdf
\else
\fi
\usepackage{epsfig}
\usepackage{epstopdf}
\usepackage{amsmath,amssymb}
\usepackage{epstopdf}
\usepackage{latexsym}
\usepackage{epsfig}
\usepackage[linesnumbered,ruled]{algorithm2e}
\usepackage{algpseudocode}
\usepackage{subcaption}
\usepackage{float}
\usepackage{amsfonts}
\usepackage{amsthm}
\usepackage[hyphens]{url}
\usepackage{bm}
\usepackage{textcomp}
\usepackage{dsfont}
\usepackage{footnote,footmisc}
\usepackage{enumitem}
\usepackage{caption}
\usepackage{array}
\usepackage{graphicx}
\usepackage{color}

\newtheorem{theorem}{Theorem}

\newtheorem{lemma}{Lemma}

\begin{document}
%
\title{Leveraging Multi-Connectivity for Multicast Video Streaming}

\author{Sadaf~ul~Zuhra\IEEEauthorrefmark{1},~\IEEEmembership{Member,~IEEE,}
        Prasanna~Chaporkar\IEEEauthorrefmark{1},~\IEEEmembership{Member,~IEEE,}
        Abhay~Karandikar\IEEEauthorrefmark{1},~\IEEEmembership{Member,~IEEE,}
        and~Pranav~Jha\IEEEauthorrefmark{1},~\IEEEmembership{Member,~IEEE}
        \vspace{-0.7cm}
\thanks{\IEEEauthorrefmark{1} 
Sadaf ul Zuhra is with NEO, Inria Sophia Antipolis Mediterranée. e-mail: sadaf-ul.zuhra@inria.fr. Prasanna Chaporkar, Abhay Karandikar, and Pranav Jha are with the Department of Electrical Engineering, Indian Institute of Technology Bombay.
e-mail: \{chaporkar, karandi, pranavjha\}@ee.iitb.ac.in.
Abhay Karandikar is currently the Director, Indian Institute of Technology Kanpur (on leave from IIT Bombay). e-mail: karandi@iitk.ac.in.}
}%

\maketitle

\begin{abstract}
Multi-connectivity has emerged as a key enabler for providing seamless connectivity in cellular mobile networks. However, its potential for improving the quality of multicast transmissions has remained unexplored. In this paper, we investigate the use of multi-connectivity in wireless multicast streaming. Multi-connectivity can significantly improve the performance of multicast services. It especially benefits the cell edge users who often suffer from poor channel conditions. In this work, we assess the impact of multi-connectivity on the performance of multicast streaming. We propose procedures for establishing multi-connectivity in a multicast system and address the associated resource allocation problem. We prove that the optimal resource allocation problem is NP-hard. We propose a greedy approximation algorithm for this problem and prove that no other polynomial-time algorithm can provide a better approximation. Since video streaming is the primary use case under consideration here, we use traces from actual videos to generate realistic video traffic patterns in our simulations. Our simulation results clearly establish that multi-connectivity results in considerable performance improvement in multicast streaming.
\end{abstract}

\begin{IEEEkeywords}
Multicast, MBMS, Multi-connectivity, LTE, 5G, Video streaming
\end{IEEEkeywords}

\IEEEpeerreviewmaketitle

\section{Introduction}\label{sec:intro}
Cellular mobile networks today need to cater to an extremely high density of users and provide them with high quality video streaming~\cite{5g_basic}. Comprising nearly $80 \%$ of the total data traffic~\cite{ericsson}, video has become one of the major drivers of innovation in cellular mobile networks. Networks are constantly adapting to meet the increasing demand for high quality video streaming. Developing techniques that optimize resource utilization is essential to successfully meet the resource requirements of these bandwidth intensive services. In this paper, we focus on wireless multicast transmission and multi-connectivity as means for meeting these requirements. 
\par 
Multicast refers to one to many transmission in which a base station can transmit common content to multiple users simultaneously. It enables serving several users who require the same content, on the same resources. This is especially useful for transmitting video streaming content like live telecasts of sports events, movie premiers, political events, news telecasts. Using multicast for such services can save considerable network resources and serve a large number of users in a limited bandwidth~\cite{self,tech_report}.
\par
Multi-connectivity (MC) allows users to potentially connect to and receive content from multiple base stations and over multiple Radio Access Technologies (RATs) simultaneously. Multi-connectivity is expected to be a key enabler in Third Generation Partnership Project (3GPP) Fifth Generation (5G) wireless networks~\cite{multilink}. The high data rate, ultra-reliable low latency, and high mobility requirements of 5G necessitate the reduction of radio link failures due to mobility. Multi-connectivity makes it possible to avoid such failures and ensures seamless connectivity for mobile users~\cite{mobility}.
By allowing users to receive content from multiple base stations simultaneously, it allows serving a larger number of users and improves the performance of cell edge users.
\par 
Using multi-connectivity with multicast transmission, advantages of both these techniques can be pooled to provide significant performance improvements in video streaming. 
Using multi-connectivity with multicast significantly increases the serving capacity of a cell, reduces the dependence of multicast transmissions on the weakest users in the system, and makes the multicast operations more robust. 
Multi-connectivity has received considerable attention from the research community in the past few years especially for throughput and handover improvement~\cite{dual_lte,dual_1,dual_handover,dual_hetnet} but its use with multicast transmissions has not been considered. 
\par
In this paper, we explore the use of multi-connectivity in multicast transmissions. We define procedures for establishing multi-connectivity for users in a multicast system. Since a multi-connected system involves users receiving content from multiple base stations simultaneously, the corresponding resource allocation problem needs to consider a global view of the system to make optimal allocation decisions. We formulate the resource allocation problem for this system with the objective of maximizing the total number of users served. We prove that this resource allocation problem is NP-hard and propose a centralized greedy approximation algorithm for solving it. Since centralized resource allocation incurs additional control overheads, we also propose a distributed allocation policy. We evaluate the performance improvements provided by multi-connectivity in multicast through extensive simulations. 
\par
 Throughout this paper, we discuss the procedures and solutions in the context of Multimedia Broadcast Multicast Services (MBMS) of 3GPP Long Term Evolution (LTE). This is because the standardization for MBMS in 5G has not yet been completed~\cite{5g_enhancements}. The proposed procedures can be easily extended to a 5G system since the basic features of MBMS services have been adopted in 5G as well~\cite{5g_enhancements}. Similarly, the discussions on resource allocation hold true for both LTE and 5G systems.

\subsection{Related Literature}
Multi-connectivity plays a pivotal role in enhancing system capacity, improving reliability, avoiding radio link failures, and reducing outage probabilities~\cite{mobility}. As a result, it has become an essential component of the next generation of wireless mobile networks~\cite{3gpp_5gdc} that demand increased reliability and low latency along with high data rates.
Dual connectivity was introduced in Release $12$~\cite{rel12} of 3GPP standard for LTE for addressing connectivity issues arising due to mobility of users. Mobility improvement continues to be one of the primary focuses of multi-connectivity research for 5G networks. In this section, we discuss the current state of the research being carried out in various aspects of multi-connectivity in cellular mobile networks.
\par
A physical layer design for New Radio (NR) MBMS (NR-MBMS) has been proposed in~\cite{5g_phy} by building on the current specifications of 5G NR. Mitigating radio link failures using multi-connectivity in ultra dense deployments of intra-frequency 5G networks has been investigated in~\cite{mobility}. It is shown in~\cite{mobility} that multi-connectivity results in significant reduction in radio link failures while also improving the throughput of cell-edge users. A form of proportional fair allocation policy for multi-connected ultra dense networks has been proposed in~\cite{effective}. Under this policy, the priority of a user is determined based on load balancing and its signal characteristics.
\par
Stringent network availability is essential for Ultra Reliable Low Latency (URLLC) applications in 5G. In~\cite{urllc}, it has been shown through extensive numerical and simulation based analysis that multi-connectivity improves the network availability for URLLC applications. It also results in better utilization of system resources in ultra-dense 5G networks through load-aware cell selection~\cite{load}.
\par
Various architectures have been proposed for implementing multi-connectivity in 5G~\cite{multi_icc}. Throughput performance of multi-connectivity in distributed and cloud-based heterogeneous network architectures has been compared in~\cite{nonideal}. Cloud-based heterogeneous networks have been shown to provide better throughput performance. In~\cite{mc_architecture}, an architecture for 5G that integrates multiple Radio Access Technologies (RATs) has been proposed. This architecture allows seamless integration of LTE and Wireless Local Area Network (WLAN) with 5G and enables seamless inter-RAT multi-connectivity. A control and user plane split architecture for multi-connectivity in 5G NR has been proposed in~\cite{cu_split}. The proposed approach does not use macro cells for user plane transmissions of multi-connected users to avoid impacting the performance of single-connected users.
\par
Compared to single connected systems, there is tremendous reduction in the transmit power required to achieve a certain outage probability and spectral efficiency when multi-connectivity is used~\cite{reliable}. Multi-connectivity has also been studied as a means of optimizing power consumption, especially for 5G heterogeneous cloud radio access networks~\cite{saimler2020multi}.
\par
 Millimeter-wave (mmWave) transmissions form a core component of 5G NR~\cite{imt} owing to their capability for providing high data rates and low latency. However, they suffer from poor propagation characteristics~\cite{aziz2016architecture} resulting in rapid channel variations and poor session continuity.
  Use of multi-connectivity along with guard bands has been shown to provide major improvement in mmWave session continuity~\cite{session_continuity}.
  In~\cite{petrov2017dynamic}, the authors put forward a methodology for performance evaluation of multi-connectivity in ultra-dense urban mmWave networks. Their evaluations reveal that multi-connectivity leads to improvements in denial of service and session drop probabilities.
  \par
  Complexity and signaling overheads are expected to go up as the number of multi-connected links increases. The tradeoff between the system complexity and performance improvement in a multi-connected mmWave system has been studied in~\cite{gapeyenko2018degree}. It has been shown that having up to $4$ simultaneous links leads to significant improvements in outage probability and spectral efficiency. Beyond this, the gains obtained are marginal.
  Link scheduling problem in multi-connected mmWave networks has been addressed in~\cite{tatino2018maximum}. The authors propose a network throughput optimizing algorithm that approaches the global optimum solution. Uplink multi-connectivity frameworks have been proposed in~\cite{giordani_efficient,giordani2016multi} that can efficiently monitor channel dynamics and link directions in mmWave transmissions, leading to efficient scheduling and session management.
 Multi-connectivity along with network coding enable transmission of high quality video streaming services over mmWave networks~\cite{drago2018reliable}. 
\par
In addition to cellular networks, multi-connectivity also finds applications in Vehicle-to-anything (V2X) services. It can play a pivotal role in fulfilling the Quality of Service (QoS) requirements in V2X services. A Radio Access Network (RAN) based solution for managing the communication interfaces of vehicles in a V2X network has been proposed in~\cite{pimrc_v2x}.
\par
Multi-connectivity finds use in a diverse set of applications as is made clear by the state-of-the-art discussed in this section. However, its use in multicast streaming has remained unexplored. In this work, we investigate the use of multi-connectivity for this important application. Our work clearly reveals the tremendous potential of multi-connectivity multicast for transforming video streaming over mobile cellular networks.

\subsection{Contributions}
The main contributions of this paper are as follows:
\begin{itemize}
\item We propose the use of multi-connectivity in wireless multicast transmissions. We asses its impact on the performance of multicast video streaming and establish the performance gains resulting from it.
 \item We define procedures for establishing multi-connectivity in the existing 3GPP multicast architecture and the associated control signaling requirements.
 \item We formulate the resource allocation problem in a multi-connected multicast system with the objective of maximizing the number of users served and prove that it is an NP-hard problem.
We, therefore, propose a greedy approximation algorithm for this problem that provides an approximation factor of $(1-1/e)$. This is the best approximation possible for the problem.
 \item Through extensive simulations, we establish the performance improvements resulting from the use of multi-connectivity in wireless multicast, particularly for video streaming. To generate realistic video traffic patterns, we use traces from actual video streams~\cite{asu_traces,asu_traces2} in all our simulations.
\end{itemize}

The rest of this paper is organized as follows. We provide an overview of the existing 3GPP standards for multicast and multi-connectivity in Section~\ref{sec:dc_multicast}. This is followed by a discussion of how these two techniques can be used together within the fourth and fifth generation wireless mobile networks in Section~\ref{sec:mc_mbms}. We define the procedures for establishing multi-connectivity in multicast transmissions in Section~\ref{sec:procedure_mc_multicast}. The system model and the associated resource allocation problem formulation are discussed in Section~\ref{sec:sys_model_dc}. In Section~\ref{subsec:nphard_mc_multicast}, we prove NP-hardness of the resource allocation problem. We present the approximation algorithm and prove its approximation ratio in Section~\ref{sec:greedy_mc_multicast}. We then examine the use of distributed resource allocation for MC multicast in Section~\ref{sec:RA_mc_multicast}. Finally, we present the simulation results in Section~\ref{sec:simulations_mc_multicast} and conclude in Section~\ref{sec:conclusions_mc_multicast}.
\par

\section{MBMS: An Overview} \label{sec:dc_multicast}
Multicast and broadcast services were standardized in Release $9$~\cite{tenth} of the 3GPP standards under the name of Multimedia Broadcast Multicast Services (MBMS).
MBMS includes two modes of multicast operation, Single Cell Point-To-Multipoint (SC-PTM) and MBMS Single Frequency Network (MBSFN). SC-PTM involves multicasting of content within a single cell whereas, in MBSFNs, all eNodeBs (eNBs) in an MBSFN area transmit the same content in strict synchronization~\cite{mbmsrel14}. MBSFN transmissions require precise synchronization between all eNBs in the MBSFN area and an extended Cyclic Prefix (CP) to enable User Equipments (UEs) to effectively combine the content received from multiple eNBs. The extended CP reduces the system throughput and the need for tight synchronization between eNBs results in significant control overheads. Various enhanced versions of MBMS have been standardized in later releases of 3GPP standards~\cite{kaliski2019further}. Further enhanced MBMS (FeMBMS) introduces, among other things, new features for handling video, larger bandwidth allocation for MBMS, and multi-cell connectivity for MBMS services~\cite{kaliski2019further}.

\subsection{MBMS Control Signaling}
MBMS, as standardized by 3GPP, is an idle mode procedure~\cite{3gpp_idle}. This means that, there is no need to establish a Radio Resource Control (RRC) connection for a UE to receive MBMS services. Most of the control information relating to MBMS operations is carried on a separate logical channel, the Multicast Control Channel (MCCH)~\cite{3gpp_36331}. The only MBMS related information sent over the Broadcast Control Channel (BCCH) is the information needed by UEs to acquire the MCCH(s). This information is carried by the MBMS specific SystemInformationBlock, {\it SystemInformationBlockType13} (SIB13)~\cite{3gpp_36331}. MBMS user data is carried over Multicast Traffic Channels (MTCH). Using the information provided over the MCCH, a UE can read the MTCH corresponding to the MBMS session that it is interested in.  
\subsection{MBMS Architectural Aspects}
To support MBMS, three new network elements have been added to the LTE architecture, Broadcast Multicast Service Centre (BM-SC), MBMS GateWay (MBMS-GW) and Multicell/Multicast Coordination Entity (MCE)~\cite{mbmsrel14}. The positioning of these elements within the architecture is as shown in Figure~\ref{fig:mbms_architecture}. BM-SC serves as an interface between core network and multicast/broadcast content providers. It is responsible for transporting MBMS data into the core network, managing group memberships and subscriptions and charging for MBMS sessions~\cite{tenth}. MCE is responsible for allocating radio resources to the eNBs~\cite{tenth} MBSFN operations. MBMS-GW uses IP multicast to forward the MBMS session data to the eNBs. The eNBs can then transmit the data to the UEs via wireless multicast/broadcast.
\par 
\begin{figure}[!htb]
\centering
\includegraphics[scale = 0.6]{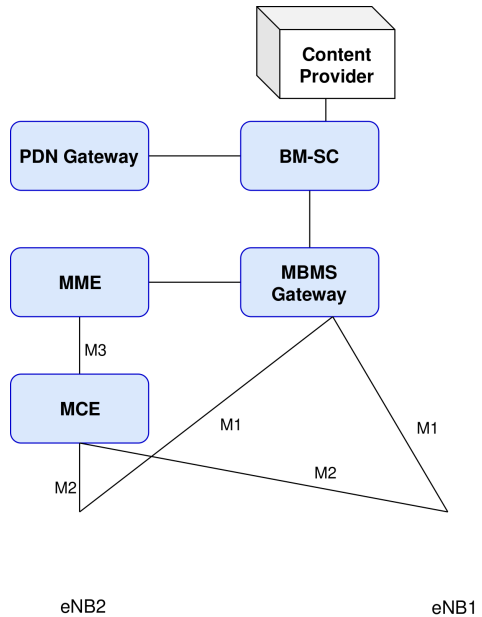}
\caption{MBMS architecture}
\label{fig:mbms_architecture}
\end{figure}


\section{Enabling Multi-connectivity in MBMS} \label{sec:mc_mbms}
MBMS user plane protocol architecture defines an additional Synchronization (SYNC) protocol layer on the transport network layer for content synchronization~\cite{3gpp_36300}. It carries additional information for identifying transmission times and detecting packet loss. The SYNC protocol is terminated in BM-SC and the eNBs. 
Since MBMS data sent to eNBs in a particular region emanates from the same BM-SC, the contents arriving at these eNBs are in sync. UEs can, therefore, receive and combine multiple copies of the same content received from these eNBs. The proposed multi-connectivity multicast exploits this inherent synchronization in MBMS systems. It enables users to obtain multicast content from multiple sources without the need for any additional synchronization.
Moreover, since MBMS is an idle mode procedure, we do not require UEs to establish an RRC connection to any eNB for using MC multicast. A UE may use MC multicast while being in RRC idle mode.
\par
We define a different dynamic between the primary and secondary eNBs of a multi-connected UE than what is traditionally defined in unicast transmissions~\cite{3gpp_5gdc}. 
Firstly, depending on its capability, we allow that a UE can connect to any number of eNBs and receive multicast content from all of them. A UE can also stay in the RRC idle mode if it is not connected to any eNB and receive content from any number of eNBs. For a UE using MC multicast in RRC idle mode, `primary' eNB refers to the eNB that it is camped on. For a UE in RRC connected mode, `primary' eNB refers to the eNB that it is connected to. All other eNBs that the UEs may receive content from are referred to as secondary eNBs. Secondly, in MC multicast, primary and secondary eNBs of a UE do not work in a traditional master-slave configuration. Secondary eNBs are not dictated by the primary eNB in their interaction with the UE. A multicast UE can receive relevant control information and multicast data from multiple eNBs independent of each other. As such, there is no real distinction between `primary' and `secondary' eNBs for a UE. Each eNB that serves the UE under MC multicast is equivalent for the UE. Note that, we use the terms `primary' and `secondary' eNB in various places in this paper for the ease of distinguishing between various eNBs that a UE is receiving data from. 
\par
MC multicast has a potential to provide all the benefits of MBSFN transmissions in a considerably simpler framework. Like in MBSFNs, a UE can receive multicast content from a number of eNBs, resulting in improved Signal-to-Noise Ratio (SNR), especially for the cell edge users. However, unlike MBSFN operations, eNBs are not required to use the same  Physical Resource Blocks (PRBs) 
for streaming the multicast content. In MC multicast, the same MBMS services are streamed through multiple eNBs and each eNB allocates PRBs to the multicast streams independently. Each eNB can, therefore, optimize the resource allocation to various services in its cell. The resulting frequency diversity significantly improves the probability of reliable reception of the MBMS content. A multicast UE can choose to decode one of the multiple copies of the content received by it. As we shall see in Section~\ref{sec:simulations_mc_multicast}, this results in significant performance improvement of multicast operations. 

\section{Procedures for Establishing Multi-Connectivity in Multicast Transmissions} \label{sec:procedure_mc_multicast}
In this section, we propose the procedures for establishing multi-connectivity in MBMS. We define MC multicast as a user initiated mechanism. As discussed in Section~\ref{sec:dc_multicast}, a UE needs to acquire the MBMS specific SIB, SIB13 and MCCH from an eNB to begin receiving MBMS session content from it. The procedures for establishing multi-connectivity for UEs in RRC connected and RRC idle modes vary in certain signaling aspects. We explain each of these procedures below.
\begin{enumerate}
 \item {\bf RRC Idle mode}:
 A UE in RRC idle mode is informed of the available MBMS sessions by its primary cell that it is camped on. If the UE is interested in receiving an MBMS session and is capable of multi-connectivity, it can choose to receive the content from multiple eNBs in its vicinity. If the UE chooses to receive the session from multiple eNBs, it starts listening to the broadcast channels of its neighboring eNBs. It receives MasterInformationBlock (MIB), SIB1, SIB13 of the neighboring cells. SIB13 obtained from the eNBs contains the MBMS relevant information of these eNBs. The UE can then receive the content from any number of these eNBs where the MBMS session of its interest is available.
 \par
 To start receiving the session content, the UE reads MCCH(s) of these eNBs. MCCH contains the information needed by the UE to obtain the relevant MTCH(s). This procedure is illustrated in Figure~\ref{fig:procedure_idle_sc}. The UE thus receives multiple copies of the same multicast content over MTCHs of multiple eNBs, resulting in improvement of the received SNR. It should be noted here that the UE does not need to establish an RRC connection to any of the eNBs in this procedure. 
 
 \begin{figure}[t!]
\centering
\includegraphics[scale = 0.5]{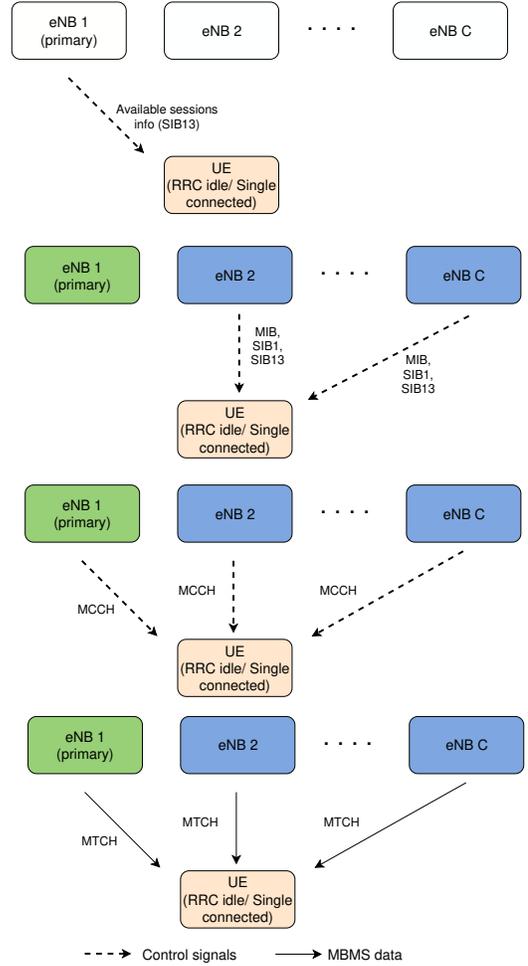}
\caption{Procedure for establishing multi-connectivity multicast for UEs in RRC idle mode and single connected UEs.}
\label{fig:procedure_idle_sc}
\end{figure}

\item {\bf RRC Connected mode}:
A UE may have already established an RRC connection for some unicast service by the time an MBMS session starts. After a UE establishes an RRC connection to a cell, it stops listening to the broadcast channels of other cells. When such a UE is informed of an MBMS session that it is interested in, it can choose to either receive the content only from the cells that it is connected to or from multiple cells using MC multicast. The procedure for establishing multi-connectivity for such a UE will be different depending upon whether the UE is single connected or dual connected. The procedures for both these cases are defined as follows. \par
 a) {\it Single Connected UE}:
 A single connected UE is notified of the available MBMS sessions by the primary eNB that it is connected to. If the UE is interested in an MBMS session, it can either choose to receive the corresponding content from its primary cell alone or use MC multicast to receive it from multiple sources. In case the UE decides to receive the MBMS content from its primary eNB, it reads the corresponding MCCH and receives the content over the relevant MTCH of its primary eNB. If the UE chooses to use MC multicast instead, it starts listening for broadcast information from its neighboring cells. It acquires MIB, SIB1 and SIB13 of these cells. It then acquires MCCH(s) of the additional cells where the session of its interest is available. MCCH provides the UE with the allocation information needed to read MTCH of its interest. This procedure is illustrated in Figure~\ref{fig:procedure_idle_sc}. Note that, the UE does not establish an RRC connection to any of the secondary eNBs. \par

b) {\it Dual Connected UE}:
Consider a UE that is dual connected and receiving some unicast service from two different cells. This UE can choose to receive MBMS content using MC multicast in one of the following ways. 
\par
\ \ i. {\it From the primary eNB alone}: In this case, the UE acquires MCCH and the relevant MTCH from its primary eNB and no additional procedures are required. \par
\ \ ii. {\it From its primary eNB and some other eNBs in its neighborhood}: This is the same as the case of a single connected UE in 2a above and the same procedures apply. \par
\ \ iii. {\it From the two eNBs it is dual connected to}: If the MBMS session of interest is available in the secondary cell of the UE, it can choose to receive MBMS content from the same two eNBs that it is dual connected to. In the existing 3GPP standards for dual connectivity, the primary eNB acts as the control plane anchor for the UE~\cite{3gpp_5gdc}. All the control information from the secondary eNB is transmitted to the UE via the primary eNB. Therefore, we propose that, SIB13 from the secondary eNB is also transmitted over the X2 interface to the primary eNB instead of the UE having to listen for it separately. It can then acquire MCCH and relevant MTCH from both the cells independently. \par
\ \ iv. {\it From the two eNBs it is connected to as well as other eNBs in its neighborhood}: This is a combination of ii. and iii. above and the same procedures are followed. Figure~\ref{fig:procedure_dual} illustrates this procedure.
\end{enumerate}

\begin{figure}[t!]
\centering
\includegraphics[scale = 0.5]{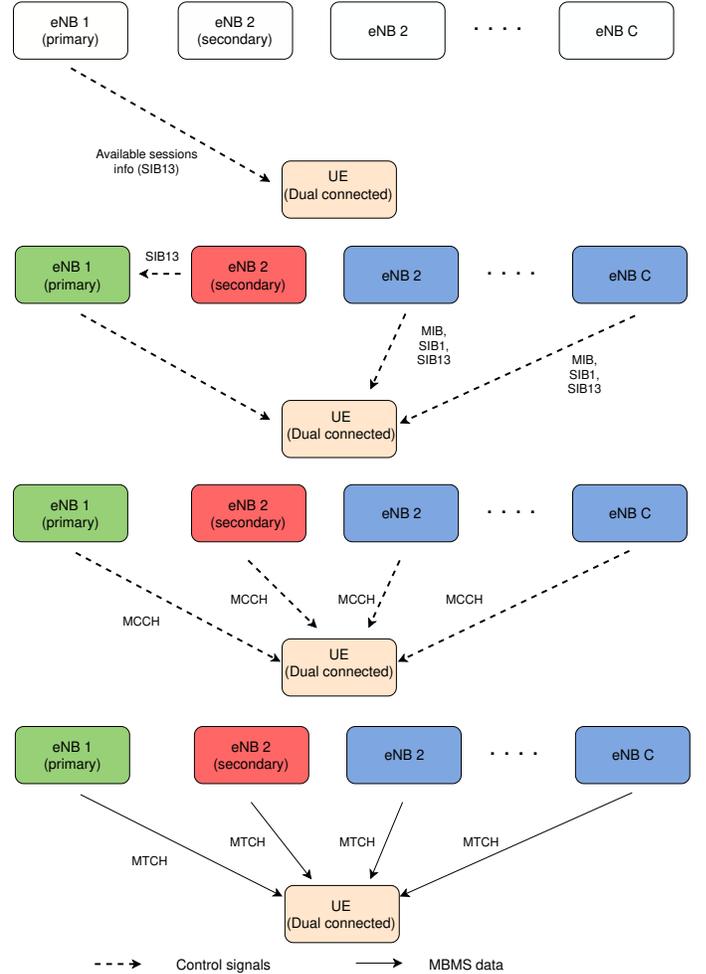}
\caption{Procedure for enabling multi-connectivity multicast for dual connected UEs.}
\label{fig:procedure_dual}
\end{figure}
In the next section, we discuss the resource allocation problem in an MC multicast system.

\section{Resource Allocation in MC Multicast} \label{sec:sys_model_dc}
We formulate the MC multicast resource allocation problem with the aim of maximizing the number of multicast users successfully served in the system. Before defining the problem, we discuss the details of our system model.
\subsection{System Model} \label{subsec:sys_model_mc_multicast}
We consider a network of $C$ cells. Each cell has an eNB located at the center. There are $M$ multicast UEs in the system that are all capable of multi-connectivity and can potentially be served by any number of eNBs. There is a multicast session available for streaming in all the cells. The multicast stream has a certain required rate $R$ at which the content needs to be streamed to the subscribed UEs. Multicast content is streamed at this rate $R$ whenever the multicast session is active. All multicast UEs are subscribed to the ongoing multicast session. The UEs in a cell that are subscribed to the same multicast session form a single multicast group and receive the streaming content over the same PRBs. 
The UEs can potentially receive the multicast streaming content from any number of neighboring eNBs in addition to their primary eNB. A multi-connected UE, therefore, belongs to multiple multicast groups streaming the same content. We assume that the multicast stream in a cell is allocated one PRB in each sub-frame.
 \par
 Resource allocation to various multicast streams can either be done by each eNB independently or by a central controller that manages the eNBs. The multicast data stream in the primary and secondary cells of a UE may or may not be scheduled on the same PRB. 
\par
 The channel states of UEs vary across time and frequency. As a result, a UE experiences different channels in different sub-frames and across different PRBs in a sub-frame. Depending on the channel state of a UE, there's a certain maximum rate it can successfully decode in a PRB. Since the multicast content is transmitted at rate $R$, a UE may or may not be successfully served by the eNB that it is connected to. For instance, say cell $c$ is streaming the multicast content over PRB $j$ in sub-frame $t$. Let $r^c_{jk}[t]$ be the maximum decodable rate for UE $k$ in PRB $j$ of cell $c$ in sub-frame $t$. If $R > r^c_{jk}[t]$, UE $k$ will not be able to successfully receive the content from cell $c$. On the other hand, if $R \leq r^c_{jk}[t]$, UE $k$ will be successfully served by $c$.
  A multi-connected UE successfully receives data in sub-frame $t$ if it can decode the content from any of the eNBs that it receives the content from. On the other hand, a UE that is not multi-connected would successfully receive data only if it can decode the content from its primary cell. We now discuss and formally define the resource allocation problem for this system. 

\subsection{Problem Formulation} \label{subsec:problem}
The resource allocation problem in an MC multicast system is aimed at serving as many UEs successfully in a sub-frame as possible. Since a multi-connected UE can receive streaming content from multiple eNBs, its performance depends on the PRB allocation in multiple cells. Therefore, we must optimize over all the cells in a region since allocation of resources in individual cells is not optimal for a multi-connected system. We now mathematically define the optimal resource allocation problem.
\par
$M$ multicast users distributed in $C$ cells can potentially receive multicast content from all the eNBs in their neighborhood. $[M]$ is the universal set of all users. There are $N$ PRBs available in each cell. Denote by $U_{jc} \subseteq [M]$, the set of users that would be successfully served if PRB $j$ is allocated to the multicast service in cell $c$. Set ${\cal U}_c = \{U_{1c},U_{2c}, \ldots, U_{Nc} \}$ is the sub-collection of such sets for cell $c$. Let ${\cal U} = \{ {\cal U}_1, \ldots , {\cal U}_C\}$. The resource allocation problem can then be stated as follows: \par
${\bf K^\star}$ :
Given the universe $[M]$ and a collection of sets ${\cal U} = \{ {\cal U}_1, \ldots , {\cal U}_C\}$, determine ${\cal U}' \subseteq {\cal U}$ such that
$|\bigcup_{U_{jc} \in {\cal U}'}U_{jc}|$ is maximized subject to $|{\cal U}'| = C$ and $|{\cal U}'\cap {\cal U}_c| =1, \ \forall \ c$.
\par
In the next section, we discuss the computational feasibility of this problem.

\section{${\bf K^\star}$ is NP-hard} \label{subsec:nphard_mc_multicast}
The optimal resource allocation problem ${\bf K^\star}$ is an NP-hard problem. We prove this by reduction from the Maximum Coverage Problem (MCP)~\cite{mcp}. MCP is a well known NP-hard problem and is defined as follows: \\
 MCP takes as input a universal set ${\cal S}$, a number $k$ and a collection of sets ${\cal T} = \{T_1,T_2, \ldots, T_m\}$ where each $T_j \subseteq {\cal S}$. The objective of MCP is to determine a sub-collection ${\cal T}' \subseteq {\cal T}$ such that
 ${\cal T}' \in \arg\max_{|{\cal T}'| \leq k} |\bigcup_{T_j \in {\cal T}'} T_j|$.

\begin{theorem} \label{theorem:np_hard}
 ${\bf K^\star}$ is an NP-hard problem.
\end{theorem}

\begin{proof}
The detailed proof is given in Section~\ref{subsec:proof_np}.
\end{proof}
Now that we have proved that the multi-connectivity resource allocation problem is NP-hard, the best we can do is construct approximation algorithms that provide some performance guarantees. We propose one such algorithm and prove its performance guarantees in the next section.

\section{Approximation Algorithm for ${\bf K^\star}$} \label{sec:greedy_mc_multicast}
We construct a Centralized Greedy Approximation (CGA) algorithm for solving ${\bf K^\star}$. The pseudo-code for this algorithm is given in Algorithm~\ref{algo:greedy}. CGA works by maximizing the number of additional users served in each iteration. In the first iteration, CGA chooses $U_{jc}$ from ${\cal U}$ that serves the maximum number of users. In the subsequent steps, it picks $U_{jc}$'s that serves the maximum number of unserved users. In each step, the set chosen is from a different sub-collection ${\cal U}_c$ i.e., $c$ in the subscript of the chosen sets is different for each set picked. The collection of sets chosen after $C$ such iterations ${\cal U}_G$, is the output of the algorithm.
\par 
In the next result, we prove that the resulting solution has an approximation factor of $\left(1-\frac{1}{e}\right)$. This means that the solution provided by this approximation algorithm serves at least $\left(1-\frac{1}{e}\right)$ of the number of users that would be served by the optimal algorithm.  \par

\begin{theorem} \label{theorem:approximation}
 The CGA algorithm (Algorithm~\ref{algo:greedy}) is a $\left(1-\frac{1}{e}\right)$ approximation for ${\bf K^\star}$. In fact, no other algorithm can achieve a better approximation unless P = NP.
 \end{theorem}

\begin{algorithm}
	\KwIn{ Universe $[M]$, ${\cal U} = \{ {\cal U}_1, \ldots , {\cal U}_C\}$, $C$ }
	Initialize: ${\cal U}_G = \phi$\\
	\For{$n = 1:C$}{
	Pick $U_{j^\star c^\star} \in {\cal U}$ that covers the maximum number of elements from $[M] \setminus \bigcup_{U_{jc} \in {\cal U}_G}U_{jc}$ \\
	${\cal U}_G \leftarrow {\cal U}' \bigcup \{U_{j^\star c^\star}\}$ \\
	${\cal U} \leftarrow {\cal U} \setminus {\cal U}_{c^\star}$
	}
	\caption{Greedy Approximation Algorithm for ${\bf K^\star}$}
	\label{algo:greedy}
\end{algorithm}

Let $OPT$ denote the number of UEs that would be served by the optimal solution. Let $m_n$ be the total number of UEs served by CGA up to and including the $n^{th}$ iteration. $b_n = OPT - m_n$ is the difference between the number of UEs served by the optimal algorithm and the number of UEs served by CGA up to the $n^{th}$ iteration. Note that $m_0 = 0, b_0 = OPT$ and $m_C$ is the total number of UEs served by CGA. 
\par
In order to determine the approximation ratio of CGA, we first prove the following two results that will eventually help us determine the approximation ratio in Theorem~\ref{theorem:approximation}.

\begin{lemma} \label{lemma:1}
 $m_{n+1}-m_n \geq \frac{b_n}{C}$.
\end{lemma}
\begin{proof}
 The proof is given in Section~\ref{proof:lemma1}.
\end{proof}

\begin{lemma} \label{lemma:2}
 $b_{n+1} \leq \left(1-\frac{1}{C}\right)^{n+1} OPT$.
\end{lemma}
\begin{proof}
 The proof is given in Section~\ref{proof:lemma2}.
\end{proof}
We can now prove Theorem~\ref{theorem:approximation}. 
\begin{proof}
 From Lemma~\ref{lemma:2}, 
 \begin{align*}
  &b_C \leq \left(1-\frac{1}{C}\right)^{C} OPT,\\
  &\implies OPT - m_C \leq \left(1-\frac{1}{C}\right)^{C} OPT \leq \frac{OPT}{e}, \\
  &\implies m_C \geq \left(1-\frac{1}{e} \right) OPT.
 \end{align*}
Thus, CGA provides a $\left(1-\frac{1}{e}\right)$ approximation for ${\bf K^\star}$. \par

This is the best possible approximation for ${\bf K^\star}$. If there was an algorithm that could provide a better approximation, that algorithm would also provide a better approximation for MCP because a solution for ${\bf K^\star}$ can be mapped to a solution of MCP in polynomial time using Algorithm~\ref{algo:solution_map}. This is a contradiction since the greedy algorithm is known to be the best possible approximation for MCP unless P = NP~\cite{feige}. Therefore, no other algorithm can provide a better approximation for ${\bf K^\star}$ than CGA.
\end{proof}

\section{Distributed versus Centralized Allocation} \label{sec:RA_mc_multicast}
The CGA algorithm discussed in the previous section is a centralized algorithm. It requires the presence of a central controller that can make allocation decisions based on a global view of the multicast region. In the absence of such a centralized controller, allocation decisions would be made by each cell individually based only on the knowledge of the UEs connected to it.
In such a distributed allocation, each cell allocates resources to the multicast stream independently. In a multi-connected system, this type of allocation does not fully reap the benefits of multi-connectivity. We illustrate this with the following example. Consider a $2$ cell system containing cells $c_1$ and $c_2$. There are two PRBs available for allocation in each cell. We denote these as $P_1$ and $P_2$. $c_1$ contains four users, $\{ u_1, u_2, u_3, u_4\}$ and $c_2$ has two users $\{u_5, u_6\}$. All users are subscribed to the same multicast stream. $u_1$ has a good channel only in $P_1$ and can successfully receive content only on $P_1$. Users $u_3,u_4, u_5$ and $u_6$ have a good channel only in $P_2$ and can, therefore, successfully receive content only on $P_2$. $u_2$ has a good channel in both the PRBs and would be served on either of them. Users $u_1, u_3, u_4$ are connected to both the cells and can receive content from either of them.
\par
Let us now look at the allocations that will be done by a distributed policy that maximizes the number of users served in each cell independently. $c_1$ considers the users connected to it and allots $P_2$ to the stream because it serves the maximum number of users, ($u_2,u_3,u_4$). $c_2$ also optimizes independently and allocates $P_2$ to the stream to serve ($u_3,u_4,u_5,u_6$). Under this allocation, $u_1$ remains unserved even though it was multi-connected, since it could only receive the content over $P_1$. On the other hand, $u_3$ and $u_4$ receive content from both the cells. In contrast, a centralized policy would consider users of both cells together and allocate $P_2$ to the stream in $c_2$ and $P_1$ in $c_1$ and successfully serve all users in the system.
\par
Any centralized allocation policy, even if it is sub-optimal, will always do better in terms of the number of users successfully served than a policy which allocates resources in a distributed manner. A centralized policy does not necessarily mean that the policy is optimizing over the entire system. Any form of centralization that looks beyond just the individual cell will reap a better performance than a completely uncoordinated allocation. We now define a distributed allocation policy. We use this policy for the purpose of simulations in Section~\ref{sec:simulations_mc_multicast}.

\subsection{Distributed Greedy Allocation}
In Distributed Greedy Allocation (DGA) policy, each eNB allocates resources to the multicast streams by optimizing over the users connected to it. This policy solves ${\bf K^\star}$ for each cell individually. In each sub-frame, an eNB allocates PRBs to the multicast streams such that the maximum number of users associated with it are served. When optimizing over a single cell, the problem can be solved in polynomial time. The pseudo-code for this algorithm is given in Algorithm~\ref{algo:uncoord_greedy}. Recall that the set $U_{jc}$ in Algorithm~\ref{algo:uncoord_greedy} is the set of all users that would be successfully served if PRB $j$ were allocated to the multicast stream in $c$. 
$x_{jc}$ is an indicator random variable that indicates whether or not PRB $j$ has been allocated to the multicast stream in cell $c$.

\begin{algorithm}
	\KwIn{Sets ${\cal U}_c = \{U_{1c}, \ldots, U_{Nc}\}$ for all $c \in [C]$}
	Initialize: $x_{jc} = 0$ for every $j,c$\\
	\For{$c = 1:C$}{	
		Assign $j^\star = \arg\max_j |U_{jc}|$ \\
		$x_{j^\star c} \leftarrow 1$

	}
	\caption{Distributed Greedy Allocation}
	\label{algo:uncoord_greedy}
\end{algorithm}

\section{Simulations} \label{sec:simulations_mc_multicast}
We simulate a seven cell LTE urban macro network~\cite{3gpp_5gchannel}. An eNB is located at the center of each cell and UEs are distributed uniformly at random throughout the cells. LTE specific physical layer conditions have been created using 3GPP channel models~\cite{parameters}. SNR to rate mapping has also been done according to 3GPP specifications~\cite{parameters}.
\par
 A single multicast streaming service is available in all the cells. Other relevant simulation parameters are given in Table~\ref{parameters_mc_multicast}. The cell edge users in the system are multi-connected to all the eNBs in the system. In all the cells, a PRB is allocated to the multicast stream in each sub-frame. The content corresponding to the multicast stream is transmitted at rate $R$ in the PRB allocated to it. Multi-connected users successfully receive a packet in a sub-frame if they can decode the content from at least one of the eNBs. Other users should be able to decode the content from their primary eNBs for being served. 
\par
Resource allocation to the multicast streams is done using the CGA algorithm proposed in Section~\ref{sec:greedy_mc_multicast} and the DGA policy proposed in Section~\ref{sec:RA_mc_multicast}. We use the number of packets delivered successfully and the number of UEs successfully served in a sub-frame as the performance metrics in these simulations. We compare the performance of the centralized (CGA) and distributed (DGA) resource allocation algorithms. We also compare the performance of MC with Single Connectivity (SC) and MBSFN to establish the performance gains resulting from the use of MC in multicast transmissions. For resource allocation in SC multicast, we use the distributed algorithm from Section~\ref{sec:RA_mc_multicast} where each eNB only considers the UEs in its own cell while making the allocation decisions.

 \begin{center}
\begin{table}
\captionof{table}{System Simulation parameters~\cite{3gpp_5gchannel,parameters}} \label{parameters_mc_multicast} 
\begin{center}
\begin{tabular}{ | m{3 cm} | m{5 cm}|} 
 \hline
 \textbf{Parameters} & \textbf{Values}\\
 \hline
 \hline
 System bandwidth & $20$ MHz \\
 \hline
 Cell radius & $250$ m \\ 
 \hline
 Path loss model & L = $128.1 + 37.6 \log10(d)$, $d$ in kilometers\\ 
 \hline
 Lognormal shadowing & Log Normal Fading with $10$ dB standard deviation \\
 \hline
 White noise power density & $-174$ dBm/Hz \\
 \hline
 eNB noise figure & $5$ dB \\
 \hline
 eNB transmit power & $46$ dBm \\
\hline
 \end{tabular}
 \end{center}
 \end{table}
\end{center}

In Figure~\ref{fig:dist_central}, we plot the average number of packets successfully received by UEs under the CGA and the DGA resource allocation algorithms. We transmit one packet in every sub-frame and plot the average number of packets successfully received by all the UEs in the system over a period of $10$ s ($10000$ sub-frames). We observe that CGA performs much better than the DGA policy. It succeeds in successfully serving the UEs in a significantly larger number of sub-frames. 

 \begin{figure}[t!]
        \centering
\includegraphics[scale = 0.5]{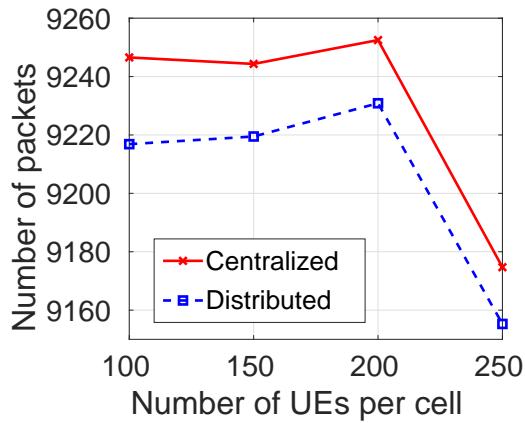} 
\caption{Average umber of packets received successfully under MC using centralized and distributed allocation}
\label{fig:dist_central}
\end{figure}

In Figures~\ref{fig:packets}a to~\ref{fig:traces}b, we compare the performance of MC multicast with that of SC multicast. From here onwards, only the CGA algorithm has been used for allocation in MC multicast. For the plots in Figure~\ref{fig:packets}a to~\ref{fig:unserved}b, data is transmitted at a fixed rate in each sub-frame. The points in these plots are obtained by averaging over $10000$ sub-frames.
\par 
Figure~\ref{fig:packets}a illustrates the number of packets successfully received under MC and SC as the number of users per cell increases. We observe a decline in the number of packets successfully received as the number of UEs increases. However, the number of packets successfully delivered under MC multicast is much larger than that under SC multicast. Figure~\ref{fig:packets}b plots the same metric as a function of cell radius. We observe a decline in the number of packets successfully received as the cell sizes increase. This is expected since the path loss of the cell edge users increases as the cells become larger. The key thing to note here is that the performance gap between MC and SC follows an increasing trend. The relative performance of MC and SC is similar to what we observe in Figure~\ref{fig:packets}a. 

\begin{figure}[t!]
	\centering{
		\begin{tabular}{cc}
			\hspace{-0.2cm}\includegraphics[width=0.45\columnwidth]{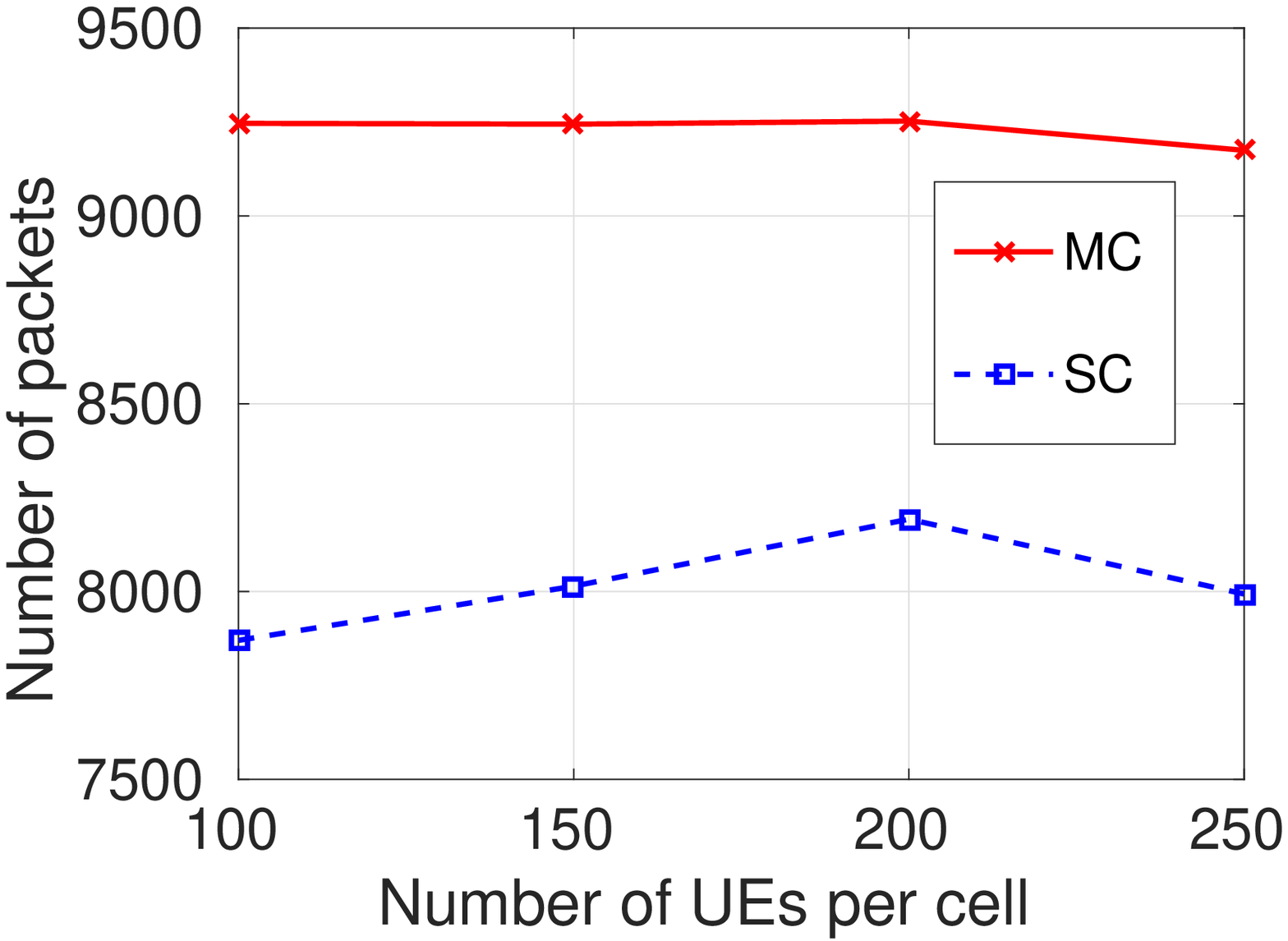}&\hspace{-0.3cm}
			\includegraphics[width=0.45\columnwidth]{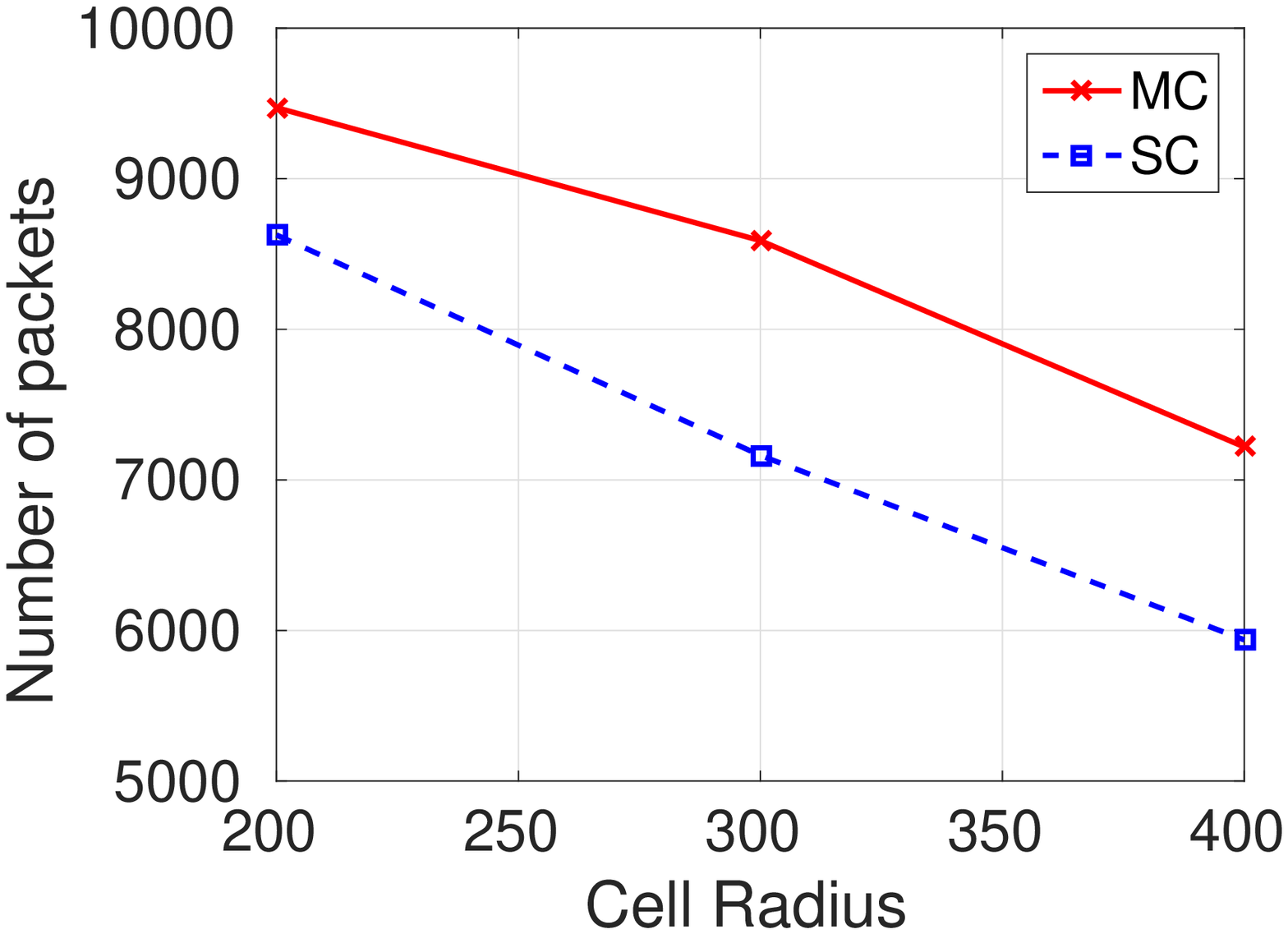}\\
			\hspace{-0.2cm}\scriptsize (a) & \scriptsize (b)
		\end{tabular} 
		\caption{{Average number of packets received successfully under greedy approximation algorithm as a function of, a)\ number of users, b)\ cell radius }}
		\label{fig:packets}
	}
\end{figure}

%
%

Figures~\ref{fig:unserved}a and~\ref{fig:unserved}b plot the average number of users left unserved in a cell per sub-frame as a function of increasing number of users and cell radius respectively.  The number of users left unserved increases as the number of users and cell radius increases. Performance gap between MC and SC multicast also increases as the number of users in each cell increases. We observe that the number of unserved UEs increases nearly three fold in the absence of multi-connectivity

.

\begin{figure}[t!]
	\centering{
		\begin{tabular}{cc}
			\hspace{-0.2cm}\includegraphics[width=0.45\columnwidth]{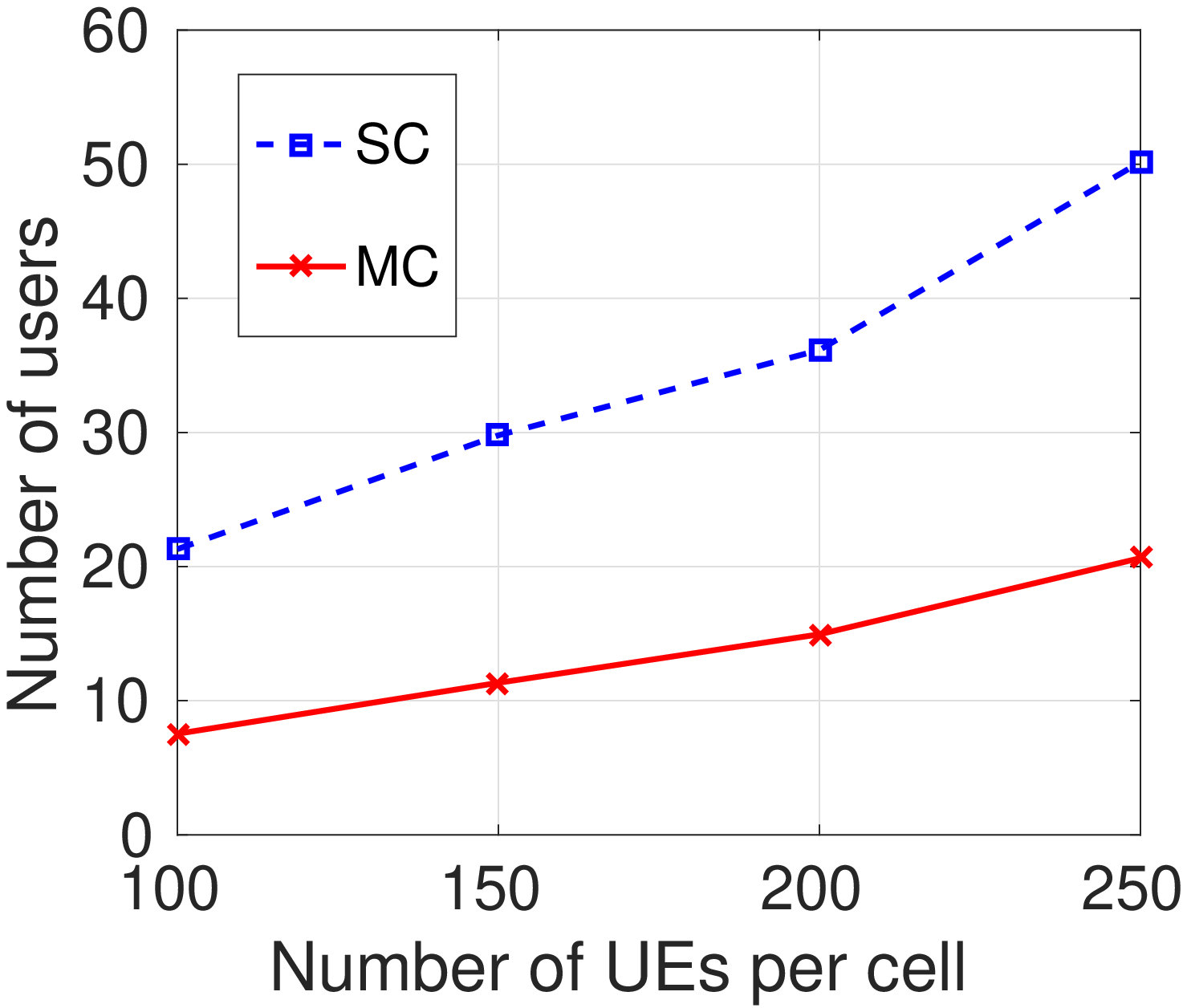}&\hspace{-0.3cm}
			\includegraphics[width=0.45\columnwidth]{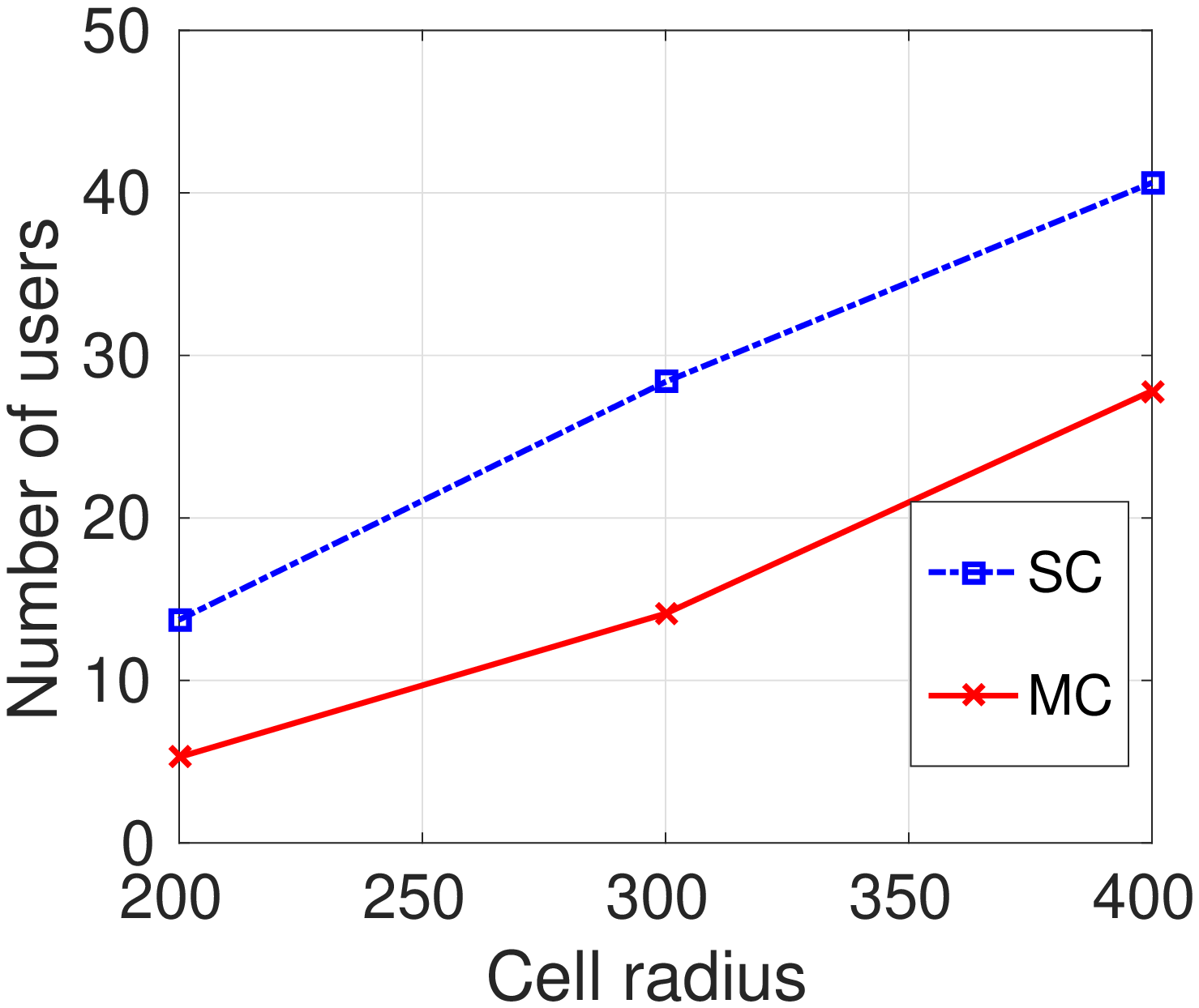}\\
			\hspace{-0.2cm}\scriptsize (a) & \scriptsize (b)
		\end{tabular} 
		\caption{{Average number of unserved users under greedy approximation algorithm as a function of, a)\ number of users, b)\ cell radius }}
		\label{fig:unserved}
	}
\end{figure}

%
%
%

In Figures~\ref{fig:traces}a and~\ref{fig:traces}b, we compare the performance of MC and SC while serving a real-time video stream. To generate realistic video traffic patterns, we use traces of a video of Tokyo Olympics (obtained from \url{http://trace.eas.asu.edu})~\cite{asu_traces}. For these simulations, the rate of transmission varies every sub-frame, according to size of the video frame being transmitted. We run the simulations for the duration of the video stream and then average the results over all the sub-frames. In Figure~\ref{fig:traces}a, we observe that MC multicast delivers significantly more packets successfully than SC multicast. From Figure~\ref{fig:traces}b, we observe that many more UEs are left unserved under SC than under MC. The performance gap between the two increases as the number of UEs in the system increases.

\begin{figure}[t!]
	\centering{
		\begin{tabular}{cc}
			\hspace{-0.2cm}\includegraphics[width=0.45\columnwidth]{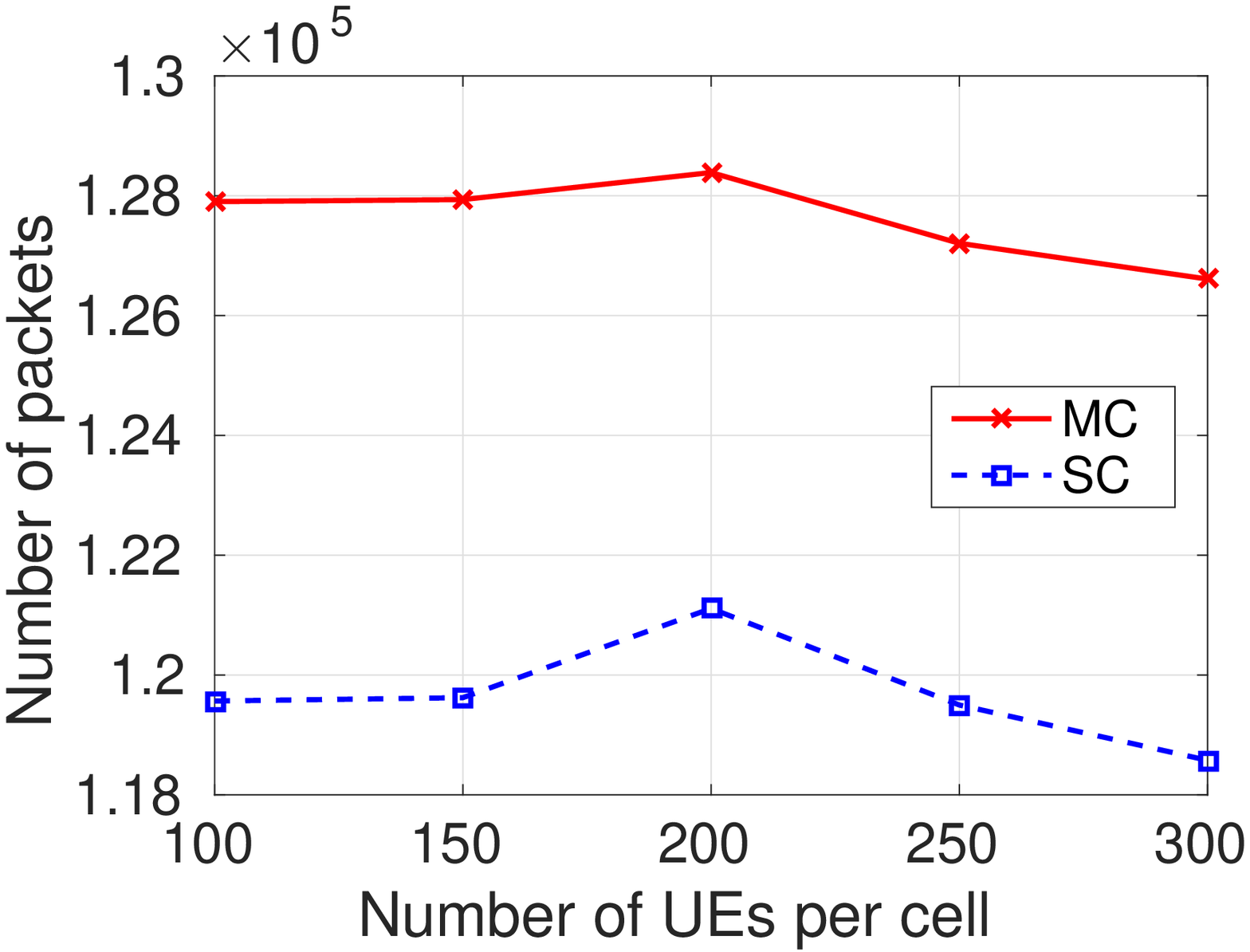}&\hspace{-0.3cm}
			\includegraphics[width=0.45\columnwidth]{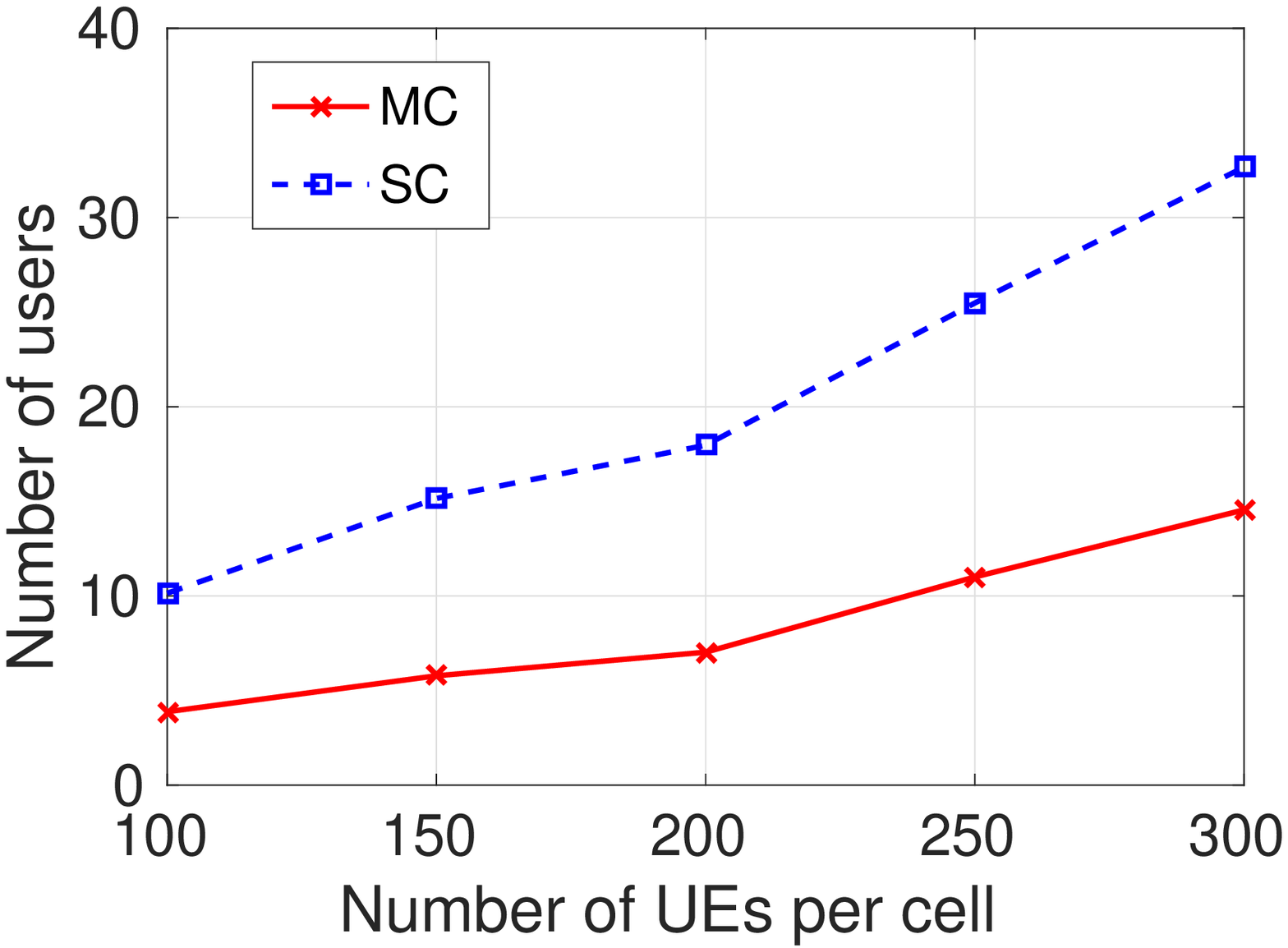}\\
			\hspace{-0.2cm}\scriptsize (a) & \scriptsize (b)
		\end{tabular} 
		\caption{{Performance comparison of MC and SC multicast for a real-time video stream, a)\ average number of packets received successfully, b)\ average number of unserved UEs}}
		\label{fig:traces}
	}
\end{figure}

%
%
%

\begin{figure}[t!]
	\centering{
		\begin{tabular}{cc}
			\hspace{-0.2cm}\includegraphics[width=0.45\columnwidth]{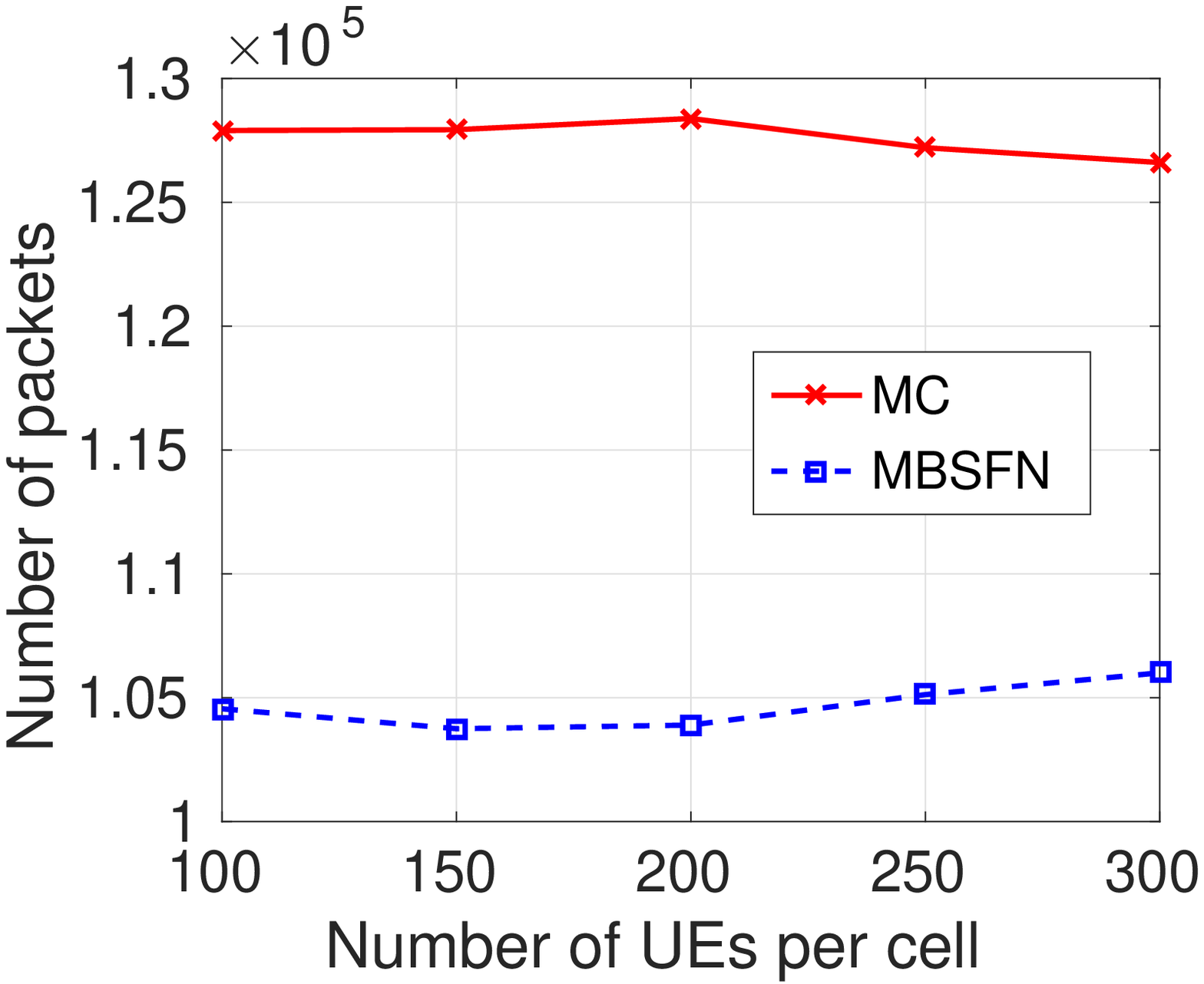}&\hspace{-0.3cm}
			\includegraphics[width=0.45\columnwidth]{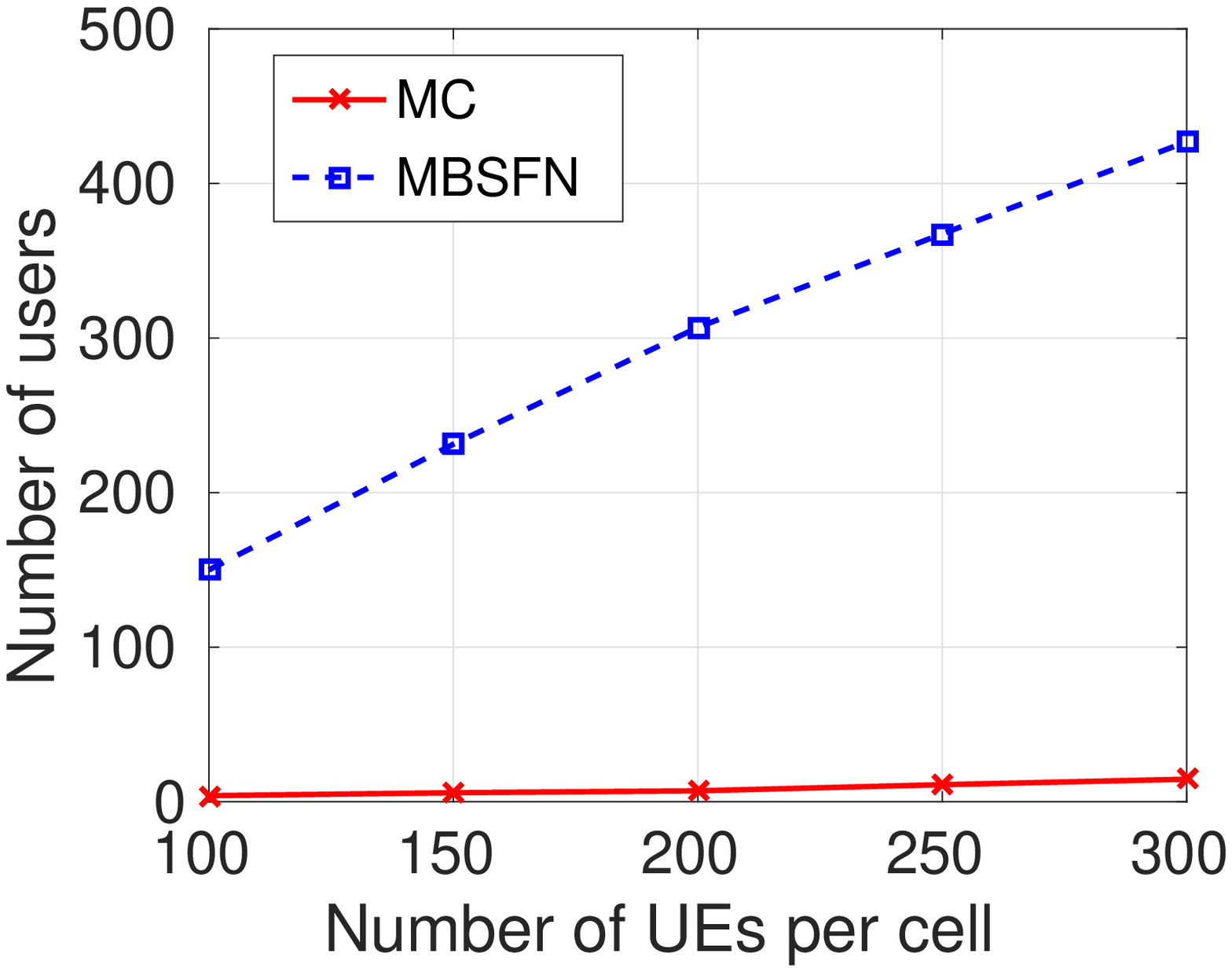}\\
			\hspace{-0.2cm}\scriptsize (a) & \scriptsize (b)
		\end{tabular} 
		\caption{{Performance comparison of MC multicast and MBSFN for a real-time video stream, a)\ average number of packets received successfully, b)\ average number of unserved UEs }}
		\label{fig:mbsfn}
	}
\end{figure}

%
%
%
In Figures~\ref{fig:mbsfn}a and~\ref{fig:mbsfn}b, we compare the performance of MC multicast with that of MBSFN transmissions. We consider that all the cells in our system belong to the same MBSFN area. MBSFN requires the multicast content to be transmitted over the same PRB by all eNBs. For resource allocation in MBSFN, we choose a PRB that serves the maximum number of UEs in the entire system. Here too, we use traces of the video of Tokyo Olympics to generate realistic video traffic patterns. We observe that MC multicast performs remarkably better than MBSFN. It succeeds in delivering a much larger number of packets successfully and is able to serve significantly more UEs than MBSFN. While many UEs remain unserved under MBSFN, nearly all of them are served under MC multicast. These results validate our claims that MC multicast can provide the benefits of MBSFNs while eliminating the need for synchronization. In fact, as observed in Figures~\ref{fig:mbsfn}a and~\ref{fig:mbsfn}b, MC multicast outperforms MBSFN by large margins.
\par
These simulation results clearly indicate that using multi-connectivity with multicast provides a significant performance enhancement in multicast systems. The flexibility of potentially receiving content from multiple eNBs results in more users being served each sub-frame. Thus, MC multicast has tremendous potential especially for use in video streaming services. It can help alleviate the burden on network resources while serving the increasing video streaming traffic. 

\section{Conclusions} \label{sec:conclusions_mc_multicast}
In this paper, we have proposed the use of multi-connectivity in multicast transmissions.
This work establishes that multi-connectivity results in significant performance improvement of multicast services.
We have proposed procedures for enabling the use of multi-connectivity in MBMS. We have formulated the problem of resource allocation in multi-connected multicast systems with the aim of maximizing the number of users served. We have proved this to be an NP-hard problem. A centralized greedy approximation algorithm that provides an approximation ratio of $\left(1-\frac{1}{e}\right)$ for this problem has also been proposed. No polynomial-time algorithm can provide a better approximation. 
Through extensive simulations, we have established that the use of multi-connectivity in multicast transmissions significantly improves the system performance. Multi-connectivity enables serving a much larger number of users at higher data rates. We have also studied the performance of multi-connectivity in serving real-time video streaming applications. To generate video specific traffic patterns in these simulations, we have used traces from actual videos~\cite{asu_traces}. We have also compared the performance of MC multicast to that of MBSFNs. Our simulation results establish that multi-connectivity outperforms MBSFNs by large margins, while eliminating the need for strict synchronization and extended cyclic prefixes.

\section{Acknowledgement}
This work has been supported by the Department of Telecommunications, Ministry of Communications, India as part of the Indigenous 5G Test Bed project.

\section{Proofs} \label{sec:proofs}

\subsection{Proof of Theorem~\ref{theorem:np_hard}} \label{subsec:proof_np}
\begin{proof}
In order to prove that ${\bf K^\star}$ is NP-hard, we first define an algorithm to reduce an instance of MCP to an instance of ${\bf K^\star}$ in polynomial time. Then, we demonstrate how a solution of ${\bf K^\star}$ can be mapped to a solution of MCP. We begin with the reduction. 
 
   \begin{algorithm}
	\KwIn{MCP with collection of sets ${\cal T} = \{T_1,T_2, \ldots, T_m\}$ and a number, $k \in \mathbb{N}$}
	\KwOut{An instance of ${\bf K^\star}$ with}
	
		 $C \leftarrow k$\\ \label{line:cells}
		 $N \leftarrow m$\\ \label{line:prbs}
		 $U_{jc} \leftarrow T_j \ \forall \ j \in \{1,2, \ldots ,m\}, c \in \{1,2,\ldots, C\}$\\
	
	\caption{Pseudo-code for reducing MCP to ${\bf K^\star}$}
	\label{algo:reduction}
\end{algorithm} 
 
The pseudo-code for reducing an instance of MCP to an instance of ${\bf K^\star}$ is given in Algorithm~\ref{algo:reduction}. We define the total number of cells $C$ to be $k$ and the number of PRBs available $N$, as $m$. The set $U_{jc}$ is set to be $T_j$ for every $c$. This reduction can be accomplished in constant time ($\mathcal{O}(C)$). We now demonstrate how a solution of ${\bf K^\star}$ can be mapped to a solution of MCP.
 
    \begin{algorithm}
	\KwIn{Solution of ${\bf K^\star}$ ${\cal U}' \subseteq {\cal U}$ such that $|{\cal U}'| = C$ and $|{\cal U}'\cap {\cal U}_c| =1, \ \forall \ c$}
	\KwOut{Solution of MCP ${\cal T}'$}
	
		$T_j \in {\cal T}'$ iff $U_{jc} \in {\cal U}'$ for some $c$
	
	\caption{Pseudo-code for mapping a solution of ${\bf K^\star}$ to a solution of MCP}
	\label{algo:solution_map}
\end{algorithm}
 
 Let us assume that there exists a polynomial time algorithm for solving ${\bf K^\star}$. Say ${\cal U}'$ is the solution of ${\bf K^\star}$. This means that $|{\cal U}'| = k$, $|{\cal U}'\cap {\cal U}_c| =1, \ \forall \ c$ and ${\cal U}'$ maximizes $|\bigcup_{U_{jc} \in {\cal U}'}U_{jc}|$. This solution can be mapped to a solution of MCP as follows: \par
 Construct set ${\cal T}'$ such that, $T_j \in {\cal T}'$ iff $U_{jc} \in {\cal U}'$. Since $|{\cal U}'| = k$, $|{\cal T}'| \leq k$. Therefore, ${\cal T}'$ is a feasible solution of MCP. The pseudo-code for this mapping is given in Algorithm~\ref{algo:solution_map}. We now need to prove that this is the optimal solution of MCP. We prove this by contradiction as follows. 
 \par
 Say that that there exists ${\cal T}'' \subseteq {\cal T}$ such that $|{\cal T}''| \leq k$ and $|\bigcup_{T_j \in {\cal T}''} T_j| > |\bigcup_{T_j \in {\cal T}'} T_j|$. We can then construct ${\cal U}''$ using ${\cal T}''$ as follows. Say ${\cal T}'' = \{T_{j_1}, \ldots ,T_{j_l}\}, l \leq k$ and say $j_1<j_2<\ldots<j_l$. Then, we can construct ${\cal U}'' = \{U_{j_11}, U_{j_22}, \ldots , U_{j_ll}, U_{1(l+1)}, \ldots , U_{1C} \}$. We have, $|{\cal U}''| = C$, $|{\cal U}''\cap {\cal U}_c| =1, \ \forall \ c$ and $|\bigcup_{U_{jc} \in {\cal U}''}U_{jc}| > |\bigcup_{U_{jc} \in {\cal U}'}U_{jc}|$ which is a contradiction to ${\cal U}'$ being the solution of  ${\bf K^\star}$. Therefore, we cannot have $|{\cal T}''| \leq k$ such that $|\bigcup_{T_j \in {\cal T}''} T_j| > |\bigcup_{T_j \in {\cal T}'} T_j|$, and ${\cal T}'$ is indeed the optimal solution of MCP. 
 \par
 
 Algorithm~\ref{algo:solution_map} maps a solution of ${\bf K^\star}$ to a solution of MCP in constant time ($\mathcal{O}(C)$ assignments).
Thus, a polynomial time solution for  ${\bf K^\star}$ results in a polynomial time solution for MCP as well. This is not possible unless P = NP. Therefore, no polynomial time algorithm exists for solving ${\bf K^\star}$ i.e., ${\bf K^\star}$ is an NP-hard problem. 
\end{proof}

\subsection{Proof of Lemma~\ref{lemma:1}} \label{proof:lemma1}
\begin{proof}
 Let $U_{OPT} = \{U_1^\star, \ldots , U_C^\star \}$ be the optimal solution. Denote by $M_n$, the set of users served at the end of the $n^{th}$ iteration of CGA and by $M_n^C$ the set of users not yet covered after the end of the $n^{th}$ iteration. We have:
 \begin{align}
 & \sum_{c=1}^C |U_c^\star \bigcap M_n^C| \geq |\bigcup_{c=1}^C (U_c^\star \bigcap M_n^C)| \geq OPT-m_n, \nonumber \\
\label{eq:inequality1} &\implies \max_c |U_c^\star \bigcap M_n^C| \geq \frac{(OPT-m_n)}{C} = \frac{b_n}{C}.
 \end{align}
Since CGA picks the set that serves the maximum possible number of yet unserved users in each iteration, we have:
\begin{equation} \label{eq:inequality2}
 m_{n+1}-m_n \geq \max_c |U_c^\star \bigcap M_n^C|.
\end{equation}
From~(\ref{eq:inequality1}) and~(\ref{eq:inequality2}), $m_{n+1}-m_n \geq \frac{b_n}{C}.$
\end{proof}

\subsection{Proof of Lemma~\ref{lemma:2}} \label{proof:lemma2}
\begin{proof}
 We prove this result by induction. For $n = 0$, we have:
 \begin{align*}
  &b_1 \leq \left(1-\frac{1}{C}\right) OPT, \\
  \implies &OPT - m_1 \leq OPT - \frac{OPT}{C}, \\
  \implies &m_1 \geq \frac{OPT}{C} = \frac{b_0}{C},
 \end{align*}
which is true (from Lemma~\ref{lemma:1}). Thus, the result holds for $n=0$. We now assume that $b_{n} \leq \left(1-\frac{1}{C}\right)^{n} OPT$ and prove that  $b_{n+1} \leq \left(1-\frac{1}{C}\right)^{n+1} OPT$. By the definition of $b_n$ and $m_n$, we have:
\begin{align}
\label{eq:step1} &b_{n+1} \leq b_n - \left(m_{n+1}-m_n\right), \\
\label{eq:step2} \implies &b_{n+1} \leq b_n - \frac{b_n}{C} = b_n\left(1-\frac{1}{C}\right), \\
\label{eq:step3} \implies &b_{n+1} \leq \left(1-\frac{1}{C} \right)^{n+1} OPT.
\end{align}
(\ref{eq:step2}) follows from (\ref{eq:step1}) by Lemma~\ref{lemma:1} and (\ref{eq:step3}) follows from (\ref{eq:step2}) by our assumption that $b_{n} \leq \left(1-\frac{1}{C}\right)^{n} OPT$. Thus, by mathematical induction, the result holds for all $n$.
\end{proof}

\bibliographystyle{IEEEtran}
\bibliography{myrefs}

\end{document}